\documentclass[letterpaper, 10 pt, conference]{ieeeconf}  

\IEEEoverridecommandlockouts                              
\usepackage[compress]{cite}
\usepackage{graphicx}
\usepackage{amsmath}
\usepackage{amsfonts}
\usepackage{amssymb}
\usepackage{cases}

\usepackage{amsthm}
\usepackage{hyperref}

\newcommand{\lf}{\left}
\newcommand{\rg}{\right}
\DeclareMathOperator*{\argmax}{\arg\max}

\newtheorem{theorem}{Theorem}
\newtheorem{lemma}{Lemma}

\theoremstyle{definition}
\newtheorem{definition}{Definition}
\newtheorem*{notation}{Notation}
\newtheorem{assumption}{Assumption}
\newtheorem*{example}{Counter-Example}

\theoremstyle{remark}

\newcommand{\bt}{\beta}
\newcommand{\ep}{\varepsilon}
\newcommand{\ta}{\tau}
\newcommand{\ka}{\kappa}
\newcommand{\sg}{\sigma}
\newcommand{\Sg}{\Sigma}

\newcommand{\R}{\mathbb{R}}

\newcommand{\Rn}{\R^n}
\newcommand{\Ru}{\R^p}
\newcommand{\Rd}{\R^q}

\newcommand{\sbs}{\subset}
\newcommand{\tms}{\times}
\newcommand{\ii}{\infty}

\newcommand{\xSig}{\mathbf{x}}
\newcommand{\traj}[2]{\xSig_{#1}^{#2}}

\newcommand{\uVals}{\mathcal{U}}
\newcommand{\dVals}{\mathcal{D}}
\newcommand{\uSig}{\mathbf{u}}
\newcommand{\dSig}{\mathbf{d}}
\newcommand{\uSigs}[1]{\mathbb{U}_{#1}}
\newcommand{\dSigs}[1]{\mathbb{D}_{#1}}
\newcommand{\dStrat}{\sg}
\newcommand{\dStrats}[1]{\Sg_{#1}}

\newcommand{\vSig}{\mathbf{v}}
\newcommand{\wSig}{\mathbf{w}}

\newcommand{\standardTrajectory}{\traj{x,t}{\uSig,\dStrat}}
\newcommand{\dotStandardTrajectory}{\dot{\xSig}_{x,t}^{\uSig,\dStrat}}

\newcommand{\firstStrat}{\alpha}
\newcommand{\secondStrat}[1]{\beta_{#1}}
\newcommand{\composedStrat}{\dStrat}

\newcommand{\switchTime}[1]{s_{#1}}
\newcommand{\switchState}[1]{z_{#1}}
\newcommand{\switchSet}[1]{\mathcal{S}_{#1}}

\newcommand{\firstTraj}[1]{\traj{x,t}{#1,\firstStrat}}
\newcommand{\secondTraj}[1]{\traj{z,s}{#1,\secondStrat{z,s}}}
\newcommand{\composedTraj}[1]{\traj{x,t}{#1,\composedStrat}}
\newcommand{\barTraj}{\traj{z,s}{\bar{\vSig},\bar{\firstStrat}}}

\newcommand{\prmStrat}[1]{\bt_{#1}^1}
\newcommand{\altStrat}[1]{\bt_{#1}^2}

\newcommand{\prmSwitchTime}[1]{s_{#1}^1}
\newcommand{\prmSwitchState}[1]{z_{#1}^1}
\newcommand{\prmSwitchSet}[1]{\mathcal{S}_{#1}^1}

\newcommand{\altSwitchTime}[1]{s_{#1}^2}
\newcommand{\altSwitchState}[1]{z_{#1}^2}
\newcommand{\altSwitchSet}[1]{\mathcal{S}_{#1}^2}

\newcommand{\tar}{r}
\newcommand{\obs}{q}
\newcommand{\prm}{r_1}
\newcommand{\alt}{r_2}

\newcommand{\tarEval}[2]{\tar \lf( #1(#2), #2 \rg)}
\newcommand{\obsEval}[2]{\obs \lf( #1(#2), #2 \rg)}
\newcommand{\prmEval}[2]{\prm \lf( #1(#2), #2 \rg)}
\newcommand{\altEval}[2]{\alt \lf( #1(#2), #2 \rg)}

\newcommand{\specialTar}{\tilde{\tar}}
\newcommand{\specialTarEval}[2]{\specialTar \lf( #1(#2), #2 \rg)}

\newcommand{\supu}[1]{\sup_{\uSig \in \uSigs{#1}}}
\newcommand{\infd}[1]{\inf_{\dStrat \in \dStrats{#1}}}

\newcommand{\JA}[2]{J_{#1}^{\mathrm{A}}[#2]}
\newcommand{\JR}[2]{J_{#1}^{\mathrm{R}}[#2]}
\newcommand{\JRA}[2]{J_{#1}^{\mathrm{RA}}[#2]}
\newcommand{\JRAA}[2]{J_{#1}^{\mathrm{RAA}}[#2]}
\newcommand{\JRR}[2]{J_{#1}^{\mathrm{RR}}[#2]}

\newcommand{\VA}[1]{V_{\mathrm{A}}[#1]}
\newcommand{\VR}[1]{V_{\mathrm{R}}[#1]}
\newcommand{\VRA}[1]{V_{\mathrm{RA}}[#1]}
\newcommand{\VRAA}[1]{V_{\mathrm{RAA}}[#1]}
\newcommand{\VRR}[1]{V_{\mathrm{RR}}[#1]}

\newcommand{\VAEval}[3]{\VA{#1}\lf(#2(#3),#3 \rg)}
\newcommand{\VREval}[3]{\VR{#1}\lf(#2(#3),#3 \rg)}

\newcommand{\targetSet}{\mathcal{R}}
\newcommand{\safeSet}{\mathcal{Q}}

\newcommand{\by}[1]{\tag*{by #1}}

\title{\LARGE \bf
Exact Decomposition of Adversarial Dual-Objective Value Functions, \\
with Applications to Optimal Drug Dosing
}

\author{Dylan Hirsch$^*$, William Sharpless$^*$, and Sylvia Herbert
\thanks{Research reported in this publication was supported by the National Institutes of Health under award number T32EB009380. The content is solely the responsibility of the authors and does not necessarily represent the official views of the National Institutes of Health.}
\thanks{Dylan Hirsch (corresponding author), William Sharpless, and Sylvia Herbert are with the Department of Mechanical and Aerospace Engineering, University of California at San Diego, 9500 Gilman Drive MC 0411, La Jolla, CA 92093.
        {\tt\small dhirsch@ucsd.edu, wsharpless@ucsd.edu, sherbert@ucsd.edu.}}%
\thanks{*Equal contribution.}
}

\begin{document}

\maketitle
\thispagestyle{empty}
\pagestyle{empty}

\begin{abstract}
Hamilton-Jacobi Reachability (HJR) is a central framework in safe control theory.
While HJR has traditionally focused on a few fundamental tasks, there is increasing interest in scaling to more complex objectives. Recent works have studied the exact decomposition of the value functions for two fundamental dual-objective tasks in the adversary-free setting.
However, not all value function decompositions in HJR remain valid with an adversary.
In this work, we develop theoretical approaches to certify that for these two composite value functions, the proposed decompositions still hold with an adversary.
Finally, we show how these results can solve issues that arise when applying HJR to optimal drug regimen design.
\end{abstract}

\section{INTRODUCTION}
Hamilton-Jacobi Reachability (HJR) is a theoretical and algorithmic framework for the control of safety-critical systems with nonlinear dynamics \cite{HJR-Survey,Mitchell-Tomlin-TAC-HJR-2005,Margellos-Lygeros-TAC-Reach-Avoid-2011,Fisac-Sastry-Reach-Avoid-2015}.
Methodologically, HJR computes a value function which encodes the optimal controller for a given safety or reachability task.
The sign of the value function indicates whether the task can be satisfied, and its magnitude quantifies the extent of success or failure (e.g. the distance by which a target will be missed).

Three standard tasks in HJR are goal-reaching (``reach'') \cite{Mitchell-Tomlin-TAC-HJR-2005}, obstacle-avoidance (``avoid'') \cite{Mitchell-Tomlin-TAC-HJR-2005}, and obstacle-avoidance until goal-reaching (``reach-avoid'') \cite{Margellos-Lygeros-TAC-Reach-Avoid-2011,Fisac-Sastry-Reach-Avoid-2015}.
Each of these tasks has a corresponding value function, and critically each of these value functions can be characterized as the viscosity solution of a certain Hamilton-Jacobi partial differential equation (HJ-PDE) \cite{Evans-Souganidis-Differential-Games-Representation-Formulas-1984,Bardi-Dolcetta-Optimal-Control}, facilitating computation.

Active research in HJR includes scaling towards more complex tasks, including dual-objective \cite{Sharpless-ICLR}, reach-avoid-stay \cite{Zheng-RAS}, sequential reach-avoid \cite{Xiang-CDC-2025,chen2025controlsynthesismultiplereachavoid}, and broader signal temporal logic specifications \cite{Chen-2020-STL-Meets-Reachability,Lars,Jiang-2020,Jiang-2024}. Recently, it was demonstrated that for a broad class of specifications, the HJR value can be decomposed into a graph of values in the single-player case \cite{Sharpless-RSS}.

The work \cite{Sharpless-ICLR} in particular focuses on two fundamental dual-objective tasks: the reach-always-avoid (RAA) and reach-reach (RR) problems.
The RAA problem is similar to the standard reach-avoid problem, except that it requires the system to both reach its target without hitting the obstacle \textit{and also avoid the obstacle thereafter}.
The RR problem requires the system to reach two targets in either order.

While the standard reach, avoid, and reach-avoid problems include an adversarial player to induce robustness to disturbances, the value function decompositions in \cite{Sharpless-ICLR} were only studied in the adversary-free case, and their proofs rely critically on the ability to freely exchange the optimization over control signals with operations in the payoff functional.
In the adversarial setting, such exchanges are no longer valid due to the value function's minimax structure, which arises from the underlying game.
For this reason, other seemingly natural decompositions no longer hold 
(see section \ref{subsection:counter-example} for an intuitive counter-example).

In this work, we show that despite this obstruction, the particular decompositions introduced in \cite{Sharpless-ICLR} still hold in the adversarial case, establishing fundamental approaches for value function decomposition with an adversary.
Moreover, whereas this previous work studies a discrete-time setting, we establish these results in continuous time, where the notion of non-anticipative strategies is central \cite{elliott-kalton,Evans-Souganidis-Differential-Games-Representation-Formulas-1984}.
In proving the decompositions for these dual-objective problems, we establish fundamental approaches for value function decomposition with an adversary.

In terms of applications, while the RAA and RR problems have traditionally been explored for robotics and reinforcement learning, they also have utility more broadly.
In particular, the RAA problem is often the appropriate task for optimal dosing regimen design, in which toxic concentrations of a pharmaceutical must not occur even after their primary therapeutic goal has been achieved.
We thus provide examples demonstrating how to use our decomposition results toward such biomedical applications.

\textbf{Contributions:} (i) we theoretically certify the exactness of a certain decomposition of the RAA value function in the presence of an adversary for the continuous time, finite-horizon setting, (ii) we provide an analogous result for the RR problem, and (iii) we demonstrate via examples how to apply these results for optimal dose regimen design.

\textbf{Structure:} In section \ref{section:background}, we describe the setup of HJR.
In section \ref{section:raa-rr}, we introduce the adversarial RAA and RR problems.
In section \ref{section:theory}, we provide the decomposition theorems.
Finally, in section \ref{section:examples}, we demonstrate the utility of our theory to optimal dosing problems.

\section{HJR Setup and Background}\label{section:background}

\begin{notation}
We use the variable $x$ to represent a point in the state space, i.e. $x \in \Rn$, and we use the variable $\xSig$ to represent a state trajectory, i.e. $\xSig: [t,T] \to \Rn$.
We similarly use $u$ ($d$) for an element of the control (disturbance) space and $\uSig$ ($\dSig$) for a control (disturbance) signal.
\end{notation}

\subsection{Setup}
We consider a system with dynamics
\begin{equation}\label{eqn:dynamics}
\dot{\xSig}(\ta) = f(\xSig(\ta), \uSig(\ta), \dSig(\ta), \ta),
\end{equation}
where $f:\Rn \tms \uVals \tms \dVals \tms \R \to \Rn$, with $\uVals \sbs \Ru$ and $\dVals \sbs \Rd$.
Here, $\ta$ represents time, $\xSig$ the state, $\uSig$ the control, and $\dSig$ the disturbance. We assume the following.

\begin{assumption}\label{assumption:regularity}
    The function $f$ is continuous, and there exists some $K > 0$ such that $\|f(x,u,d,\ta)\| \le K(1 + \|x\|)$ for each $u \in \uVals$, $d \in \dVals$, and $\ta \in \R$.
    Moreover, for each $R > 0$ there exists an $L > 0$ such that $\|f(x_1,u,d,\ta) - f(x_2,u,d,\ta)\| \le L \|x_1 - x_2\|$ for all $x_1,x_2 \in \Rn$ satisfying $\max\{\|x_1\|,\|x_2\|\} \le R$ and for all $u \in \uVals$, $d \in \dVals$, and $\ta \in \R$.
\end{assumption}
\begin{assumption}\label{assumption:compactness}
    The sets $\uVals$ and $\dVals$ are compact.
\end{assumption}
We fix $T \in \R$, representing the final time.
For each initial time $t < T$, we denote by $\uSigs{t}$ the set of all measurable control signals $\uSig:[t,T) \to \uVals$ and by $\dSigs{t}$ the set of all measurable disturbance signals $\dSig:[t,T) \to \dVals$.
Conceptually, we can think of any map $\dStrat:\uSigs{t} \to \dSigs{t}$ as representing a strategy for an adversarial player that chooses the disturbance input $\dSig$ based upon the control input $\uSig$.
In HJR, we make the restriction that the adversary's strategy cannot use future information to inform current decisions.
We make this requirement precise via the notion of non-anticipativity \cite{elliott-kalton}.
\begin{definition}[Non-Anticipativity]
Let $t < T$. 
An adversary strategy $\dStrat :\uSigs{t} \to \dSigs{t}$ is \textbf{non-anticipative} if for each $s \in [t,T)$ and $\uSig_1,\uSig_2 \in \uSigs{t}$ such that $\uSig_1(\ta) = \uSig_2(\ta)$ for a.e. $\ta \in [t,s)$, we also have $\dStrat[\uSig_1](\ta) = \dStrat[\uSig_2](\ta)$ for a.e. $\ta \in [t,s)$.
\end{definition}
\noindent For each $t < T$, we denote by $\dStrats{t}$ the set of all non-anticipative $\dStrat: \uSigs{t} \to \dSigs{t}$.

For each $x \in \Rn$, $t < T$, $\uSig \in \uSigs{t}$, and $\dStrat \in \dStrats{t}$,
we let $\standardTrajectory:[t,T] \to \Rn$ be the Carath\'{e}odory solution of \eqref{eqn:dynamics} under the initial condition $\xSig(t) = x$ and the disturbance signal $\dSig = \dStrat[\uSig]$.
More explicitly, $\standardTrajectory$ is defined to be the unique absolutely continuous function from $[t,T]$ into $\Rn$ for which $\standardTrajectory(t) = x$ and $\dotStandardTrajectory(\ta) = f( \standardTrajectory(\ta), \uSig(\ta), \dStrat[\uSig](\ta), \ta)$ for a.e. $\ta \in (t,T)$.
Note that the existence and uniqueness of this solution follow from Assumptions \ref{assumption:regularity} and \ref{assumption:compactness} \cite{Friedman-Differential-Games}.
Conceptually, $\standardTrajectory$ represents the state trajectory that results from a control signal $\uSig$ and a non-anticipative adversary strategy $\dStrat$, when the system begins in state $x$ at time $t$.

\subsection{The reach, avoid, and reach-avoid tasks in HJR}
We will be concerned with three following qualitative tasks from HJR.
In the following, we assume the system is initialized in state $x$ at time $t$.
\begin{itemize}
    \item \textbf{Reach task:} regardless of the adversary's actions, the system should be in a (relatively open) target set $\targetSet \sbs \Rn \tms (-\ii,T]$ \textit{at some time} between $t$ and $T$, i.e. for all $\dStrat \in \dStrats{t}$, there is a $\uSig \in \uSigs{t}$ s.t. for some $\ta \in [t,T]$ we have $(\standardTrajectory(\ta),\ta) \in \targetSet$.
    \item \textbf{Avoid task:} regardless of the adversary's actions, the system should remain within a (relatively open) safe set $\safeSet \sbs \Rn \tms (-\ii,T]$ \textit{at all times} between $t$ and $T$, i.e.
    for all $\dStrat \in \dStrats{t}$ there exists a $\uSig \in \uSigs{t}$ s.t. for all $\ta \in [t,T]$ we have $(\standardTrajectory(\ta), \ta) \in \safeSet$.
    \item \textbf{Reach-avoid task:} regardless of the adversary's actions, the system should be in $\targetSet$ \textit{at some time} between $t$ and $T$ while remaining in $\safeSet$ \textit{at all times} between $t$ and when it reaches $\targetSet$, i.e. for all $\dStrat \in \dStrats{t}$ there is a $\uSig \in \uSigs{t}$ s.t. for some time $\ta \in [t,T]$ we have $(\standardTrajectory(\ta),\ta) \in \targetSet$ and also $(\standardTrajectory(\ka),\ka) \in \safeSet$ for all $\ka \in [t,\ta]$.
\end{itemize}
Note that above, the target (safe) sets are specified as a subset of space-time to allow for moving targets (constraints).

In HJR, these tasks, which have binary outcomes, are converted into ones with continuous outcomes by representing the set $\targetSet$ ($\safeSet$) as the strict zero super-level set of some continuous functions $\tar$ ($\obs$) and subsequently defining certain payoff functionals and value functions. We refer to \cite{HJR-Survey} for details.
For the theoretical results in this work, we will only be concerned with the value functions themselves, but for motivation it suffices to know that a task can (cannot) be successfully completed from an initial state $x$ and initial time $t$ if the corresponding value function is strictly positive (negative) when evaluated at $x$ and $t$.
The magnitude of the value function measures success or failure extent.

\subsection{The reach, avoid, and reach-avoid payoffs and values}
We now specify the payoffs and value functions for the reach (R), avoid (A), and reach-avoid (RA) tasks.
Given $x \in \Rn$, $t < T$, and $\tar,\obs: \Rn \tms (-\ii, T] \to \R$, we define the payoffs $\JR{x,t}{\tar},\JA{x,t}{\obs},\JRA{x,t}{\tar,\obs}:\uSigs{t} \tms \dStrats{t} \to \R$ by
\begin{align*}
    \JR{x,t}{\tar}(\uSig, \dStrat) &:= \max_{\ta \in [t,T]} \tarEval{\standardTrajectory}{\ta},
    \\
    \JA{x,t}{\obs}(\uSig, \dStrat) &:= \min_{\ta \in [t,T]} \obsEval{\standardTrajectory}{\ta},
    \\
    \JRA{x,t}{\tar,\obs}(\uSig, \dStrat) &:= \\
    \max_{\ta \in [t,T]} \min &\lf\{ \tarEval{\standardTrajectory}{\ta}, \min_{\ka \in [t,\ta]} \obsEval{\standardTrajectory}{\ka} \rg\}.
\end{align*}
Those familiar with temporal logic may note that each of these payoff functionals is related to the robustness metric of the corresponding task specification \cite{donze-and-maler,reactive-synthesis}, though we avoid temporal logic notation for unfamiliar readers.
Next, given $\tar,\obs: \Rn \tms (-\ii, T] \to \R$, we define the value functions $\VR{\tar},\VA{\obs},\VRA{\tar,\obs}:\Rn \tms (-\ii,T] \to \R$ by
\begin{align*}
    \VR{\tar}&(x,t) := \begin{cases}
        \displaystyle\infd{t}\supu{t} \JR{x,t}{\tar}(\uSig, \dStrat) & t < T, \\
        \tar(x,T) & t = T, \\
    \end{cases}
    \\
    \VA{\obs}&(x,t) := \begin{cases}
        \displaystyle\infd{t}\supu{t} \JA{x,t}{\obs}(\uSig, \dStrat) & t < T, \\
        \obs(x,T) & t = T,
    \end{cases}
\end{align*}
\begin{align*}
    \VRA{\tar,\obs}&(x,t) :=
    \begin{cases}
        \displaystyle\infd{t}\supu{t} \JRA{x,t}{\tar,\obs}(\uSig, \dStrat) & t < T, \\
        \min\{\tar(x,T),\obs(x,T)\} & t = T.
    \end{cases}
\end{align*}

\section{Composite Tasks: The RAA and RR Problems}\label{section:raa-rr}

\subsection{Payoffs and Value Functions}
Two of the fundamental composite tasks in HJR are the reach-always-avoid (RAA) and reach-reach (RR) tasks, which are mathematically complementary, application-relevant, and demonstrative of the basic theoretical tools for value function decomposition in HJR \cite{Sharpless-ICLR}.
The RAA and RR payoffs and value functions are specified as follows.
Given $x \in \Rn$, $t < T$, and $\tar,\obs,\prm,\alt: \Rn \tms (-\ii, T] \to \R$, we define the payoffs $\JRAA{x,t}{\tar,\obs},\JRR{x,t}{\prm,\alt}:\uSigs{t} \tms \dStrats{t} \to \R$ by
\begin{align*}
    \JRAA{x,t}{\tar,\obs}(\uSig, \dStrat) & := \min \lf\{ \JR{x,t}{\tar}(\uSig, \dStrat), \JA{x,t}{\obs}(\uSig, \dStrat) \rg\},\\
    \JRR{x,t}{\prm,\alt}(\uSig, \dStrat)  & := \min \lf\{ \JR{x,t}{\prm}(\uSig, \dStrat), \JR{x,t}{\alt}(\uSig, \dStrat)\rg\}.
\end{align*}

Conceptually, in the RAA task, we require the system to reach its goal while never intersecting an obstacle (even after reaching the goal, unlike in the RA task).
In the RR task, there are instead two different goals (but no obstacle), and we require the system to reach both, in either order.
For each task, the dual-objective is encoded in the corresponding payoff via the $\min$ operator, i.e. the composite payoff will be positive iff both of the basic payoffs are positive.

Given $\tar,\obs,\prm,\alt: \Rn \tms (-\ii, T] \to \R$, define the value functions $\VRAA{\tar,\obs},\VRR{\prm,\alt}:\Rn \tms (-\ii,T] \to \R$ by
\begin{align*}
    \displaystyle\VRAA{\tar,\obs}&(x,t) &:=
    \begin{cases}
        \displaystyle\infd{t}\supu{t} \JRAA{x,t}{\tar,\obs}(\uSig, \dStrat) & t < T, \\
        \min\{\tar(x,T),\obs(x,T)\} & t = T.
    \end{cases} \\
    \displaystyle\VRR{\prm,\alt}&(x,t)  &:=
    \begin{cases}
        \displaystyle\infd{t}\supu{t} \JRR{x,t}{\prm,\alt}(\uSig, \dStrat) & t < T, \\
        \min\{\prm(x,T),\alt(x,T)\} & t = T.
    \end{cases}
\end{align*}

The previous work \cite{Sharpless-ICLR} uses decompositions of these value functions into the basic value functions in the one-player case to develop algorithms for reinforcement learning. 
The key innovation in this work is that we establish that these decompositions also hold with an adversary.

We additionally note that we study these problems in the continuous-state, continuous-time, finite-horizon case, whereas in \cite{Sharpless-RSS}, they were explored in a finite-state, discrete-time, infinite-horizon scenario.
The former setup is more typical in HJR, and the latter is more typical in reinforcement learning.
While this choice is mostly stylistic in the one-player case, in two-player continuous-time settings, the non-anticipative adversary strategies are central to the theory.

\subsection{Adversary-free vs. adversarial distinction}\label{subsection:counter-example}

To motivate our results, in this subsection, we provide an informal counter-example to intuitively explain why a value function decomposition that is exact with only a controller player may no longer hold with an adversary.
This is in fact true even when there is no controller player but only a disturbance player.
Consider the payoff functional
\begin{align*}
    &J_{x,t}^\mathrm{ORR}[\prm,\alt](\uSig,\dStrat) := \\
    &\quad\max_{\ta \in [t,T]} \min \lf\{ \max_{\ka \in [t,\ta]} \prm\lf( \standardTrajectory(\ka)\rg), \max_{\ka \in [\ta,T]} \alt\lf(\standardTrajectory(\ka)\rg)\rg\},
\end{align*}
for reaching two goals in a specified order,
with corresponding value function $V_\mathrm{ORR}[\prm,\alt]$.
In the RR problem with only the controller (which attempts to maximize the payoff), the value function can be decomposed as $\VRR{\prm,\alt} = \max\{V_\mathrm{ORR}[\prm,\alt], V_\mathrm{ORR}[\alt,\prm]\}$ \cite{Sharpless-RSS}.
Intuitively, this means the controller should choose whether it is better to go to goal 1 first and goal 2 next or vice versa, and then use an $\ep$-optimal control signal for the chosen order.
However, with a disturbance (which attempts to minimize the payoff) there are scenarios in which the adversary can determine the order in which the system will reach the goals, but cannot ultimately prevent both goals from being reached.

\begin{example}
Consider the system $\dot{\mathbf{x}} = \dSig$, with $\dVals = [-1,+1]$.
Define two targets, $\targetSet_1 = \{(x,t) \in \R \tms (-\ii,T]: (x < 0.5 \text{ and } 1 < t < 2) \text{ or } (-0.5 < x \text{ and } 2 < t)\}$ and $\targetSet_2 = \{(x,t) \in \R \tms (-\ii,T]: (-0.5 < x \text{ and } 1 < t < 2) \text{ or } (x < 0.5 \text{ and } 2 < t)\}$, where $T = 3$.
Let $\prm, \alt: \R \tms (-\ii, T]$ be continuous and have strict zero super-level sets $\targetSet_1$ and $\targetSet_2$, respectively.
Note that from the initial state $x = 0$ and initial time $t = 0$, the disturbance can ensure via $\dSig \equiv -1$ that the system passes through $\targetSet_1$ and then reaches (and stays in) $\targetSet_2 \setminus \targetSet_1$, or it can ensure via $\dSig \equiv +1$ that the system passes through $\targetSet_2$ and then reaches (and stays in) $\targetSet_1 \setminus \targetSet_2$.
It cannot, however, prevent the system from eventually reaching $\targetSet_1$ and eventually reaching $\targetSet_2$.
In particular, $V_\mathrm{ORR}[\prm,\alt](0,0) \le 0$ and $V_\mathrm{ORR}[\alt,\prm](0,0) \le 0$, but $\VRR{\prm,\alt}(0,0) > 0$, i.e. the above decomposition of the RR value function breaks down for a disturbance player.
\end{example}

In the next section, we study an alternative decomposition of the RR value function that does hold with an adversary.
\section{Theoretical Results}\label{section:theory}

We first introduce the result for the RAA value function.
\begin{theorem}\label{theorem:raa-theorem}
    Suppose $\tar,\obs: \Rn \tms (-\ii, T] \to \R$ are both continuous.
    Let $\specialTar:\Rn \tms (-\ii,T] \to \R$ be given by 
    \begin{equation*}
        \specialTar(x,t) := \min \lf\{ \tar(x,t), \VA{\obs}(x,t) \rg\}.
    \end{equation*}
    Then for each $x \in \Rn$ and $t < T$, we have
    \begin{equation}\label{eqn:raa-theorem}
        \VRAA{\tar,\obs}(x,t) = \VRA{\specialTar,\obs}(x,t).
    \end{equation}
\end{theorem}
\begin{proof}
    See appendix.
\end{proof}
The practical significance of the above theorem is as follows.
In order to compute the RAA value function, $\VRAA{\tar,\obs}$, we can proceed in three steps.
First, we compute the avoid value function $\VA{\obs}$ for the obstacle, ignoring the target, using traditional HJR methods (e.g. numerical HJ-PDEs \cite{Bardi-Dolcetta-Optimal-Control,hj-py,helper-oc} or learning \cite{Reachability-RL,Deep-Reach}).
Second, we form a special target function $\specialTar$ using $\VA{\obs}$. Third, we solve for the reach-avoid value function $\VRA{\specialTar,\obs}$, again by traditional methods.
The above result guarantees that this latter value function will indeed be equal to the RAA value function.

We now turn to the result for the RR value function.
\begin{theorem}\label{theorem:rr-theorem}
    Suppose $\prm,\alt: \Rn \tms (-\ii, T] \to \R$ are both continuous.
    Let $\specialTar: \Rn \tms (-\ii, T] \to \R$ be given by
    \begin{align*}
        \specialTar(x,t) := \max \big\{ &\min\lf\{\prm(x,t), \VR{\alt}(x,t) \rg\}, \\
        &\min\lf\{\alt(x,t), \VR{\prm}(x,t) \rg\} \big\}.
    \end{align*}
    Then for each $x \in \Rn$ and $t < T$, we have
    \begin{equation}\label{eqn:rr-theorem}
        \VRR{\prm,\alt}(x,t) = \VR{\specialTar}(x,t).
    \end{equation}
\end{theorem}
\begin{proof}
    See appendix.
\end{proof}
The significance of this result is similar to the last.
First we compute $\VR{\prm}$ and $\VR{\alt}$ using standard methods.
Next, we form $\specialTar$ using $\VR{\prm}$ and $\VR{\alt}$.
Finally, we compute $\VR{\specialTar}$.
The resulting value function will by the above theorem be precisely the RR value function, $\VRR{\prm,\alt}$.

\section{Applications to Optimal Dosing}\label{section:examples}

Pharmacokinetics (PK) models describe how pharmaceutical drugs are absorbed, transported, metabolized, and excreted in the body.
Pharmacodynamics (PD) models describe how drugs affect biomolecular processes in the body.
One can use HJR to analyze a PK-PD model to help design an optimal dosing regimen.
We refer the reader to \cite{Moore_2018} for more regarding optimal control for PK-PD problems.
In practice, after an optimal dosing regimen is found with HJR, one must typically still translate the continuous optimal regimen into a discrete one (e.g. two pills per day).
In these examples, we ignore this matter and assume continuous dosing.

\subsection{Example 1: Two-compartment PK model}
In our first example, we will show how using the RAA formulation can fix a problem that can occur, even for very basic PK models, with the traditional reach-avoid formulation.
Consider a simple PK model that describes the concentration of some pharmaceutical in two different physiological compartments: the systemic circulation (i.e. blood) and the kidneys.
The pharmaceutical agent is delivered intravenously and is eliminated through the kidneys.

We would like for the time-integrated drug concentration, known in PK as the AUC, in the blood to achieve a therapeutic threshold $\theta_{\mathrm{ther}}$, while ensuring the instantaneous concentration in the kidneys remains below a toxic threshold $\theta_{\mathrm{tox}}$.
The maximal allowed supply rates of the drug at time $\ta$ is given by $\mathbf{M}(\ta)$, which is time-varying to allow us to wean the patient off treatment after some time, and monitor the patient thereafter.
This system can be modeled as follows:
\begin{align}
    \dot\xSig_1 &= \mathbf{M} \uSig - k_1 \dSig_1 \xSig_1/(\xSig_1 + K_1), \label{eqn:example-1-eq-1}
    \\
    \dot\xSig_2 &= k_1 \dSig_1\xSig_1/(\xSig_1 + K_1) - k_2 \dSig_2 \xSig_2 / (\xSig_2 + K_2),
    \label{eqn:example-1-eq-2}
    \\
    \dot\xSig_3 &=  \xSig_1.
    \label{eqn:example-1-eq-3}
\end{align}
Here, $\xSig_1$, $\xSig_2$, and $\xSig_3$ represent the concentration of drug in the blood, the concentration of drug in the kidneys, and the AUC, respectively.
The input $\uSig$ represents the rate at which the pharmaceutical is injected into the blood.
The disturbance terms $\dSig_1$ and $\dSig_2$ reflect unmodeled dynamics in the transport rate constant $k_1$ and the elimination rate constant $k_2$ respectively.
The parameters $K_1$ and $K_2$ represent half-saturation constants for transport and elimination, respectively.
We set $\uVals = [0,1]$, and $\dVals = \{d \in \R^2 \mid \|d - (1,1)\|_2 \le \ep\}$.

In Fig. \ref{fig:raa}, we compare the results of using the traditional reach-avoid approach versus the RAA approach on this optimal dosing problem.
Optimizing the control actions using the traditional reach-avoid formulation results in a rapid increase of the drug concentration in the blood, which later leads to toxicity in the kidneys (Fig. \ref{fig:raa}, dotted lines).
Even if one terminates the treatment course as soon as the therapeutic threshold is achieved, drug from the blood still eventually moves to the kidneys, again exceeding the toxic threshold (Fig. \ref{fig:raa}, dash-dot lines).
By contrast, using the RAA formulation allows one to achieve the therapeutic threshold while avoiding the toxic threshold (Fig. \ref{fig:raa}, solid lines).

\subsection{Example 2: Mutual inhibition PD model}
Consider a PD model of activity in a cell, in which two proteins $\mathrm{X}_1$ and $\mathrm{X}_2$ can bind to form an inactive complex.
We can exogenously induce expression of either protein via the inputs $\uSig_1$ and $\uSig_2$ (representing e.g. small, quickly diffusing molecules). 
This system can be modeled as follows:
\begin{align}
    \dot\xSig_1 &= \uSig_1 + \dSig_1 - k_b \xSig_1 \xSig_2 + k_u \xSig_3 - \gamma_1 \xSig_1, \label{eqn:example-2-eq-1}
    \\
    \dot\xSig_2 &= \uSig_2 + \dSig_2 - k_b \xSig_1 \xSig_2 + k_u \xSig_3 - \gamma_2 \xSig_2, \label{eqn:example-2-eq-2}
    \\
    \dot\xSig_3 &=k_b \xSig_1 \xSig_2 -  k_u \xSig_3 - \gamma_3 \xSig_3, \label{eqn:example-2-eq-3}
\end{align}
where $\xSig_1$, $\xSig_2$, and $\xSig_3$, represent the concentrations of unbound $\mathrm{X}_1$, unbound $\mathrm{X}_2$, and $\mathrm{X}_1-\mathrm{X}_2$ complex, respectively.
The disturbances $\dSig_1$ and $\dSig_2$ represent model uncertainty, with $\dVals = \{d \in \R^2 : \|d - (0,0)\|_2 \le \ep\}$.
The parameters $k_b$ and $k_u$ are binding and unbinding rate constants.
The goal is to eventually reach an $\mathrm{X}_1$ concentration above some threshold $\theta_{\mathrm{ther},1}$ and to eventually reach an $\mathrm{X}_2$ concentration above $\theta_{\mathrm{ther},2}$ as well.
In Figure \ref{fig:rr}, observe that solving a basic reach problem to simultaneously maximize the concentrations does not achieve the objective.
On the other hand, solving the RR problem allows the controller to temporally coordinate $\mathrm{X}_1$ and $\mathrm{X}_2$ production to achieve the task.
\begin{figure}[!ht]
    \vspace{2mm}   
    \centering
    \includegraphics[width=.99\columnwidth]{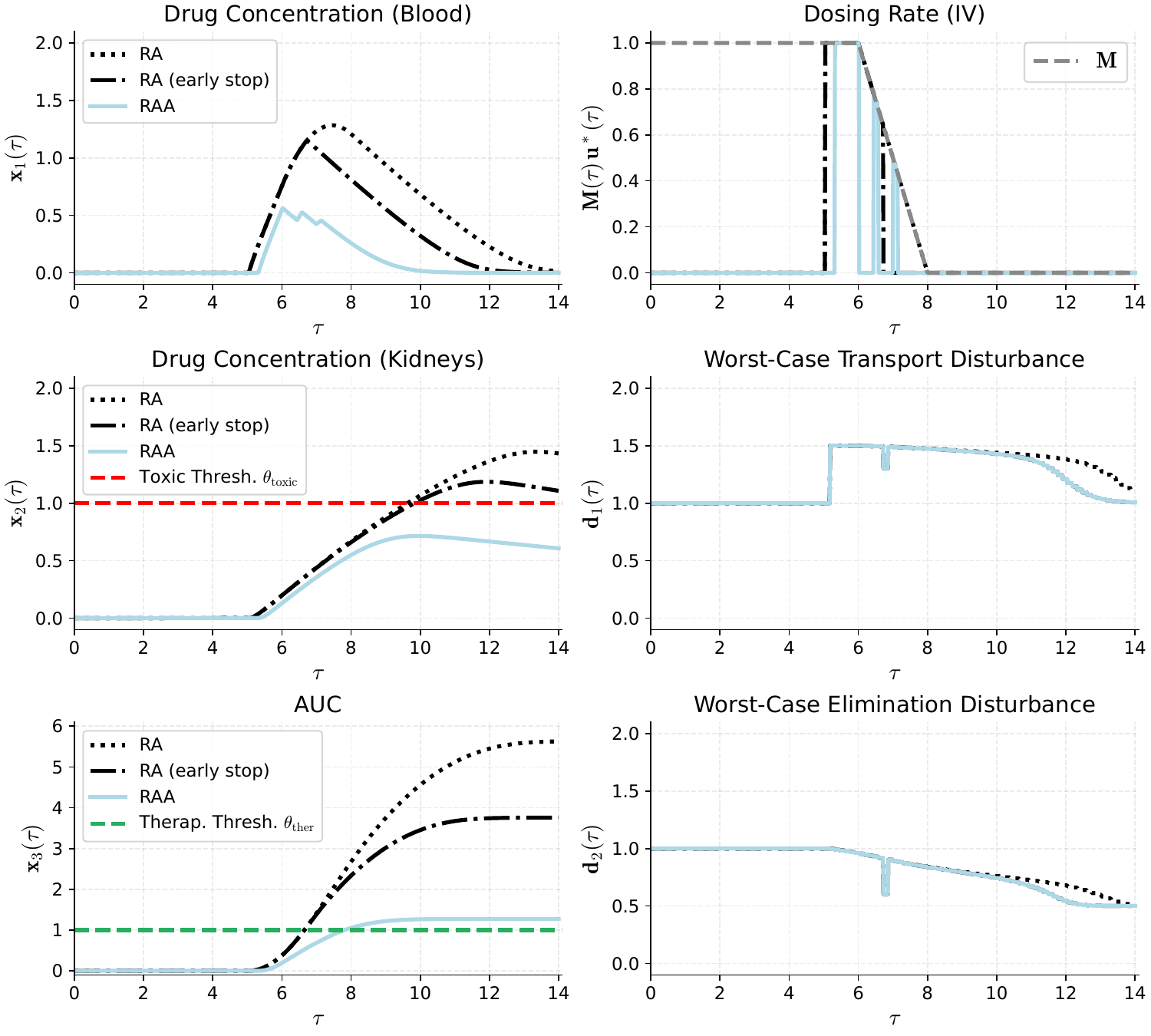}
    \caption{Closed-loop simulation of the model of \eqref{eqn:example-1-eq-1}-\eqref{eqn:example-1-eq-3} under controllers derived using the RA (dotted line) and RAA (solid line) valued functions.
    In the dash-dot line, the control is set to $0$ after the therapeutic threshold is reached (the disturbance continues to act optimally, now based on the avoid value function for zero control).
    Parameters used were $k_1 = 0.2$, $k_2 = 1$, $K_1 = 0.1$, $K_2 = 10$, $\theta_{\mathrm{ther}} = \theta_{\mathrm{toxic}} = 1$, $\ep = 0.5$.
    The target function is $\tar(x_1,x_2,x_3,t) = x_3 - \theta_{\mathrm{ther}}$, and the obstacle function is $\obs(x_1,x_2,x_3, t) = \theta_{\mathrm{toxic}} - x_2$.
    $\mathbf{M}(\tau)$ is $1$ until $\ta = 6$, $0$ after $\tau = 8$, and linear in between.
    In the RAA case, the optimal control and disturbance are chosen using the gradient of the RAA value function $\VRAA{\tar,\obs}$ prior to the first time $s$ for which $\tar(\xSig(s),s) \ge \VA{\obs}(\xSig(s),s)$ (this is the proper switching time suggested by the dynamic programming argument used in the proof of Theorem \ref{theorem:raa-theorem}) and is chosen using the avoid value function $\VA{\obs}$ thereafter.
    }
    \label{fig:raa}
\end{figure}

\begin{figure}[!ht]
    \centering
    \includegraphics[width=.95\columnwidth]{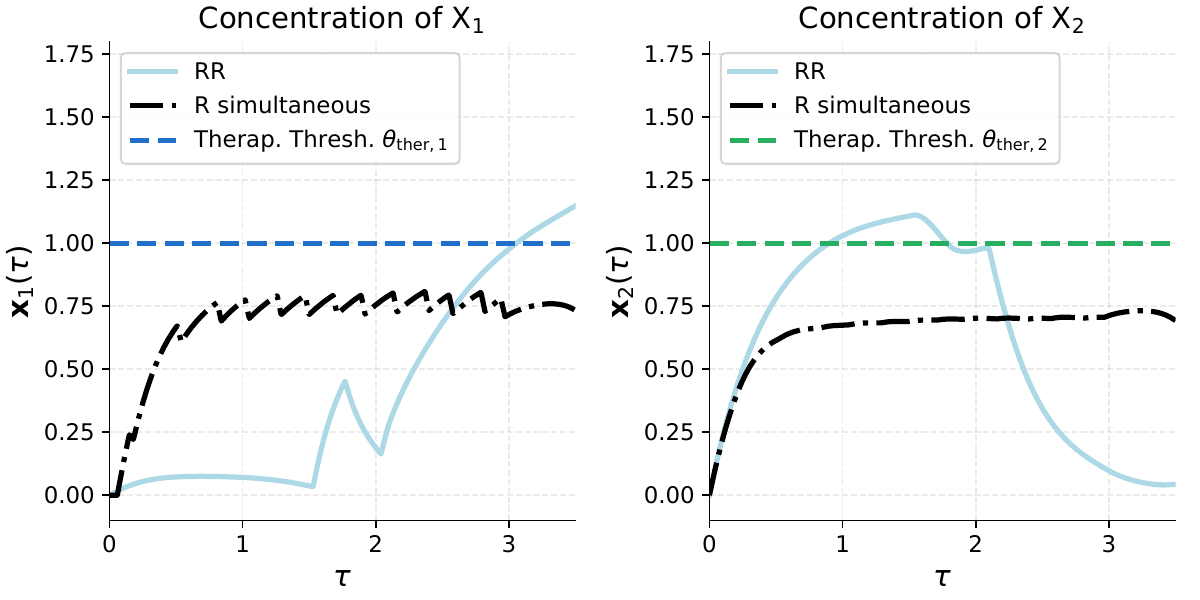}
    \caption{Closed-loop simulation of \eqref{eqn:example-2-eq-1}-\eqref{eqn:example-2-eq-3}, under controllers derived using the R (dotted line) and RR (solid line) valued functions.
    Parameters used were $k_b = 2.5$, $k_u = 0.5$, $\gamma_1 = 2$, $\gamma_2 = 2.1$, $\gamma_3 = 1.5$, $\theta_{\mathrm{ther}_1} = \theta_{\mathrm{ther}_2} = 1$, $\ep = 0.5$.
    The target function is $\tar_1(x_1,x_2,x_3,t) = x_1 - \theta_{\mathrm{ther},1}$, $\tar_2(x_1,x_2,x_3,t) = x_2 - \theta_{\mathrm{ther},2}$.
    In the RR case, the optimal control and disturbance are chosen using the gradient of the RR value function $\VRR{\tar_1,\tar_2}$ prior to the first time $s$ for which $\tar_1(\xSig(s),s) \ge \VR{\tar_2}(\xSig(s),s)$ or $\tar_2(\xSig(s),s) \ge \VR{\tar_1}(\xSig(s),s)$ (per the Theorem \ref{theorem:rr-theorem} proof), and is chosen using the reach value function $\VR{\tar_1}$ and $\VR{\tar_2}$ thereafter respectively.
    }
    \label{fig:rr}
\end{figure}
\section{Conclusions}
In this work, we demonstrated that there are certain decompositions of the value functions for the RR and RAA problems that hold in the presence of an adversary.
We also demonstrated applications of these results to optimal dosing problems in pharmacology.
Future work will focus on applying the foundational approach developed for these two tasks for decomposing more general composite value functions, i.e. where the payoff is the signal temporal logic robustness metric for a more general task \cite{donze-and-maler}. 
Namely, we aim to generalize the single-player value decomposition algebra of \cite{Sharpless-RSS}, however, it remains to be shown which rules hold in the adversarial setting.
Additionally, while the proofs of the RR and RAA value decompositions suggest the proper ``switching conditions'' between the controllers for the decomposed value functions, these conditions may produce undesirable results if the value functions are inaccurate, as is common in reinforcement learning.
Much practical and theoretical work will be required to resolve this issue.
\appendix
We will use throughout the standard fact that if the functions $\tar, \obs: \Rn \tms (-\ii,T] \to \R$ are continuous, so are $\VR{\tar}$, $\VA{\obs}$, and $\VRA{\tar,\obs}$ by Assumptions \ref{assumption:regularity} and \ref{assumption:compactness}
\cite{Bardi-Dolcetta-Optimal-Control}.

Note that, with careful handling of the subtle relationship between satisfaction sets and the value function under an adversary (see the remark after Theorem 2.3 in \cite{Cardaliaguet_1996}), one could alternatively adapt the set-based proof in \cite{chen2025controlsynthesismultiplereachavoid} to establish the $(\ge)$ direction of our Theorem \ref{theorem:raa-theorem} proof.
However, the critical step in proving our theorem is the opposite direction, where we synthesize multiple non-anticipative strategies.
The $(\ge)$ direction in the Theorem \ref{theorem:rr-theorem} proof also requires additional techniques, as there is no target ordering.
\begin{proof}[Proof of Theorem \ref{theorem:raa-theorem}]
    Let $\tar$, $\obs$, and $\specialTar$ be as in the hypothesis (recall $\tar$ and $\obs$ are assumed continuous, so $\specialTar$ is as well).
    We fix $x \in \Rn$ and $t \le T$.
    First observe that
    \begin{align*}
        &\VRAA{\tar,\obs}(x,T)
         = \min \{\tar(x,T), \obs(x,T)\}
         \\
         & = \min \{\min\{ \tar(x,T), \obs(x,T)\}, \obs(x,T)\} = \VRA{\specialTar,\obs}(x,T),
    \end{align*}
    so it suffices to assume that $t < T$.
    We show \eqref{eqn:raa-theorem} holds.

    \noindent($\le$)
    This direction proceeds by building a near-optimal composite strategy ($\composedStrat$) from a primary strategy ($\firstStrat$) and a family of secondary strategies ($\secondStrat{z,s}$).
    
    Fix $\ep > 0$.
    Select an adversary strategy $\firstStrat \in \dStrats{t}$ such that
    \begin{equation}\label{eqn:raa-le-first-strat}
        \textstyle\supu{t} \JRA{x,t}{\specialTar,\obs}(\uSig, \firstStrat) \le \VRA{\specialTar, \obs}(x,t) + \ep.
    \end{equation}
    For each $z \in \Rn$ and $s \in [t,T)$, choose $\secondStrat{z,s} \in \dStrats{s}$ such that
    \begin{equation}\label{eqn:raa-le-second-strat}
        \textstyle\supu{s} \JA{z,s}{\obs}\lf(\uSig, \secondStrat{z,s} \rg) \le \VA{\obs}(z,s) + \ep.
    \end{equation}
    For each $\uSig \in \uSigs{t}$, we define a set of candidate switch times
    \begin{equation}\label{eqn:raa-le-switch-set}
        \switchSet{\uSig} := \lf\{ s \in [t,T) :  \tarEval{\firstTraj{\uSig}}{s} \ge \VAEval{\obs}{\firstTraj{\uSig}}{s}\rg\},
    \end{equation}
    we set the actual switch time as $\switchTime{\uSig} := \min \switchSet{\uSig}$ if $\switchSet{\uSig} \ne \varnothing$ and $\switchTime{\uSig} := T$ otherwise, and we set the switch state as $\switchState{\uSig} := \firstTraj{\uSig}(\switchTime{\uSig})$.
    We define the composite strategy $\composedStrat \in \dStrats{t}$ by for each $\uSig \in \uSigs{t}$ letting $\composedStrat[\uSig] \in \dSigs{t}$ be given by
    \begin{equation*}
        \composedStrat[\uSig](\ta) := \begin{cases}
            \firstStrat[\uSig](\ta)
            & \ta < \switchTime{\uSig}
            \\
            \secondStrat{\switchState{\uSig}, \switchTime{\uSig}}[\uSig \vert_{ [\switchTime{\uSig}, T)}](\ta) & \ta \ge \switchTime{\uSig}.
        \end{cases}
    \end{equation*}
    
    We show that $\composedStrat$ is indeed non-anticipative, i.e. $\composedStrat \in \dStrats{t}$.
    Suppose $\uSig_1$ and $\uSig_2$ agree a.e. on $[t,b)$ for some $b \in [t,T)$.
    First, assume that $\min\{\switchTime{\uSig_1},\switchTime{\uSig_2}\} \le b$.
    We show $\switchTime{\uSig_1} = \switchTime{\uSig_2}$ and $\switchState{\uSig_1} = \switchState{\uSig_2}$.
    Assume $\switchTime{\uSig_1} \le \switchTime{\uSig_2}$.
    Then $\switchTime{\uSig_1} \le b$, so that 
    \begin{equation}\label{eqn:raa-le-gamma-non-anticipative}
        \firstTraj{\uSig_1}(\switchTime{\uSig_1}) = \firstTraj{\uSig_2}(\switchTime{\uSig_1})
    \end{equation} by non-anticipativity of $\firstStrat$.
    Thus, because  $\switchTime{\uSig_1} \in \switchSet{\uSig_1}$ (as $\switchTime{\uSig_1} \le b < T$), we also have $\switchTime{\uSig_1} \in \switchSet{\uSig_2}$.
    But since $\switchTime{\uSig_2}$ is the minimal element of $\switchSet{\uSig_2}$ and $\switchTime{\uSig_1} \le \switchTime{\uSig_2}$, we indeed have  $\switchTime{\uSig_1} = \switchTime{\uSig_2}$, and also
    $\switchState{\uSig_1} = \switchState{\uSig_2}$ by \eqref{eqn:raa-le-gamma-non-anticipative}.
    The same result follows if instead $\switchTime{\uSig_2} \le \switchTime{\uSig_1}$.
    
    Thus, either $\switchTime{\uSig_1} > b$ and $\switchTime{\uSig_2} > b$, in which case $\composedStrat[\uSig_1]$ and $\composedStrat[\uSig_2]$ agree a.e. on $[t,b)$ by non-anticipativity of $\firstStrat$, or $\switchTime{\uSig_1} = \switchTime{\uSig_2} \le b$ and $\switchState{\uSig_1} = \switchState{\uSig_2}$, in which case $\composedStrat[\uSig_1]$ and $\composedStrat[\uSig_2]$ agree a.e. on $[t, \switchTime{\uSig_1})$ by non-anticipativity of $\firstStrat$ and also agree a.e. on $[\switchTime{\uSig_1},b)$ by non-anticipativity of $\secondStrat{\switchState{\uSig_1}, \switchTime{\uSig_1}} = \secondStrat{\switchState{\uSig_2}, \switchTime{\uSig_2}}$.
    It follows that $\composedStrat \in \dStrats{t}$.

    Next, choose a control signal $\vSig \in \uSigs{t}$ such that
    \begin{equation}\label{eqn:raa-le-vsig}
        \textstyle\JRAA{x,t}{\tar,\obs}(\vSig, \composedStrat) \ge \supu{t} \JRAA{x,t}{\tar,\obs}(\uSig, \composedStrat) - \ep.
    \end{equation}
    For convenience, let $s := \switchTime{\vSig}$ and $z := \switchState{\vSig}$.
    By definition,
    \begin{equation}\label{eqn:raa-le-z}
        z = \firstTraj{\vSig}(s).
    \end{equation}

    There are two cases: $s < T$ and $s = T$.
    First suppose that $s < T$.
    Since $s \in \switchSet{\vSig}$ by definition, we have from \eqref{eqn:raa-le-switch-set} that $\tar(z,s) \ge \VA{\obs}(z,s)$,
    which implies
    \begin{equation}\label{eqn:raa-le-from-switching-condition}
        \specialTar(z,s) = \VA{\obs}(z,s).
    \end{equation}
    Let $\bar{\vSig} := \vSig \vert_{[s,T)}$.
    Then for each $\ta \in [t,T]$ we have,
    \begin{equation*}
        \composedStrat[\vSig](\ta) = \begin{cases}
            \firstStrat[\vSig](\ta) & \ta < s, \\
            \secondStrat{z,s}[\bar{\vSig}](\ta) & \ta \ge s.
        \end{cases}
    \end{equation*}
    Observe that for all $\ta \le s$,
    \begin{equation}\label{eqn:raa-le-gamma-v-trajectory-before-s}
        \composedTraj{\vSig}(\ta) = \firstTraj{\vSig}(\ta),
    \end{equation}
    and for all $\ta \ge s$,
    \begin{equation}\label{eqn:raa-le-gamma-v-trajectory-after-s}
        \composedTraj{\vSig}(\ta) = \secondTraj{\bar{\vSig}}(\ta).
    \end{equation}
    It follows that
    \begin{align*}
        &\VRAA{\tar,\obs}(x,t) \le  \JRAA{x,t}{\tar, \obs}(\vSig, \composedStrat) + \ep \by{\eqref{eqn:raa-le-vsig}}
        \\
        &\le \JA{x,t}{\obs}(\vSig, \composedStrat) + \ep
        \\
        &= \min \lf\{
        \min_{\ka \in [t,s]} \obsEval{\firstTraj{\vSig}}{\ka}, \JA{z,s}{\obs}(\bar{\vSig}, \secondStrat{z,s}) 
        \rg\} + \ep
        \by{\eqref{eqn:raa-le-gamma-v-trajectory-before-s},\eqref{eqn:raa-le-gamma-v-trajectory-after-s}}
        \\
        &\le \min \lf\{
        \min_{\ka \in [t,s]} \obsEval{\firstTraj{\vSig}}{\ka}, \VA{\obs}(z,s) 
        \rg\} + 2\ep \by{\eqref{eqn:raa-le-second-strat}}
        \\
        &= \min \lf\{ \specialTar(z,s), \min_{\ka \in [t,s]} \obsEval{\firstTraj{\vSig}}{\ka} \rg\} + 2\ep
        \by{\eqref{eqn:raa-le-from-switching-condition}}
        \\
        &\le \JRA{x,t}{\specialTar,\obs}(\vSig, \firstStrat) + 2\ep \by{\eqref{eqn:raa-le-z}}
        \\
        &\le \VRA{\specialTar, \obs}(x,t) + 3\ep
        \by{\eqref{eqn:raa-le-first-strat}}.
    \end{align*}

    Now instead suppose that $s = T$.
    Then $\switchSet{\vSig} = \varnothing$ by definition of $s$, so that by \eqref{eqn:raa-le-switch-set}, for all $\ta \in [t,T]$ we have $\tarEval{\firstTraj{\vSig}}{\ta} \le \VAEval{\obs}{\firstTraj{\vSig}}{\ta}$. Letting
    $\ta^* \in \argmax_{\ta \in [t,T]} \tarEval{\firstTraj{\vSig}}{\ta}$,
    we in particular have
    $\tarEval{\firstTraj{\vSig}}{\ta^*} \le \VAEval{\obs}{\firstTraj{\vSig}}{\ta^*}$, so that
    \begin{equation}\label{eqn:raa-le-T-from-switching-condition}
        \tarEval{\firstTraj{\vSig}}{\ta^*} = \specialTarEval{\firstTraj{\vSig}}{\ta^*}.
    \end{equation}
    Also, by definition of $\ta^*$,
    \begin{equation}\label{eqn:raa-le-T-tau-star}
    \textstyle\tarEval{\firstTraj{\vSig}}{\ta^*} = \max_{\ta \in [t,T]} \tarEval{\firstTraj{\vSig}}{\ta}.
    \end{equation}
    In addition, by our construction of $\composedStrat$, we have
    \begin{equation}\label{eqn:raa-le-gamma-is-alpha}
        \composedStrat = \firstStrat.
    \end{equation}
    It thus follows that
    \begin{align*}
        &\VRAA{\tar,\obs}(x,t)
        \le \JRAA{x,t}{\tar, \obs}(\vSig, \composedStrat) + \ep \by{\eqref{eqn:raa-le-vsig}}
        \\
        &= \JRAA{x,t}{\tar, \obs}(\vSig, \firstStrat) + \ep \by{\eqref{eqn:raa-le-gamma-is-alpha}}                              \\
        &= \min \lf\{ \tarEval{\firstTraj{\vSig}}{\ta^*}, \min_{\ka \in [t,T]} \obsEval{\firstTraj{\vSig}}{\ka} \rg\} + \ep
        \by{\eqref{eqn:raa-le-T-tau-star}}
        \\
        &\le \min \lf\{ \specialTarEval{\firstTraj{\vSig}}{\ta^*}, \min_{\ka \in [t,\ta^*]} \obsEval{\firstTraj{\vSig}}{\ka} \rg\} + \ep
        \by{\eqref{eqn:raa-le-T-from-switching-condition}}
        \\
        &\le \JRA{x,t}{\specialTar,\obs}(\vSig,\firstStrat) + \ep
        \\
        &\le \VRA{\specialTar,\obs}(x,t) + 2\ep 
        \by{\eqref{eqn:raa-le-first-strat}}.
    \end{align*}

    \noindent($\ge$) In this direction, we build a near-optimal control signal ($\wSig$) from a primary signal ($\vSig$) and a secondary signal ($\bar{\vSig}$) for responding to a near-optimal adversary strategy ($\firstStrat$).
    
    Fix $\ep > 0$.
    First, choose $\firstStrat \in \dStrats{t}$ such that
    \begin{equation}\label{eqn:raa-ge-first-strat}
        \textstyle\supu{t} \JRAA{x,t}{\tar,\obs}(\uSig, \firstStrat) \le \VRAA{\tar,\obs}(x,t) + \ep.
    \end{equation}
    Next, choose $\vSig \in \uSigs{t}$ and $s \in [t,T]$ such that
    \begin{align}\label{eqn:raa-ge-v-and-s}
        \min&\textstyle\lf\{ \specialTarEval{\firstTraj{\vSig}}{s}, \min_{\ka \in [t,s]} \obsEval{\firstTraj{\vSig}}{\ka} \rg\} >
        \\
        &\textstyle\supu{t} \JRA{x,t}{\specialTar,\obs}(\uSig, \firstStrat) - \ep.\nonumber
    \end{align}
    Because the left-hand-side of \eqref{eqn:raa-ge-v-and-s} is continuous in $s$ and the inequality is strict, it suffices to assume $s < T$.
    Let
    \begin{equation}\label{eqn:raa-ge-z}
        z := \firstTraj{\vSig}(s).
    \end{equation}

    For each $\uSig \in \uSigs{s}$, let $\uSig_\vSig \in \uSigs{t}$ be given by
    \begin{equation*}
        \uSig_\vSig(\ta) := \begin{cases}
            \vSig(\ta) & \ta < s,   \\
            \uSig(\ta) & s \le \ta.
        \end{cases}
    \end{equation*}
    Define the adversary strategy $\bar{\firstStrat}: \uSigs{s} \to \dSigs{s}$ by $\bar{\firstStrat}[\uSig] := \firstStrat[\uSig_\vSig]|_{[s,T)}$.
    Because $\firstStrat$ is non-anticipative, so is $\bar{\firstStrat}$, i.e. $\bar{\firstStrat} \in \dStrats{s}$.
    Finally, select $\bar{\vSig} \in \uSigs{s}$ such that
    \begin{equation}\label{eqn:raa-ge-vbar}
        \textstyle\JA{z,s}{\obs}(\bar{\vSig}, \bar{\firstStrat}) \ge \supu{s} \JA{z,s}{\obs}(\uSig, \bar{\firstStrat}) - \ep.
    \end{equation}
    Define the control signal $\wSig \in \uSigs{t}$ by
    \begin{equation*}
        \wSig(\ta) := \begin{cases}
            \vSig(\ta)       & \ta < s    \\
            \bar{\vSig}(\ta) & s \le \ta.
        \end{cases}
    \end{equation*}
    We then have $\bar{\firstStrat}[\bar{\vSig}] = \firstStrat[\wSig] \vert_{[s,T)}$.
    By \eqref{eqn:raa-ge-z},
    \begin{equation}\label{eqn:raa-ge-z-in-terms-of-w}
        z = \firstTraj{\wSig}(s).
    \end{equation}
    Also note that for all $\ta \in [s,T]$, we have
    \begin{equation}\label{eqn:raa-ge-second-trajectory}
        \barTraj(\ta) = \firstTraj{\wSig}(\ta).
    \end{equation}
    It follows that
    \begin{align*}
         &\VRAA{\tar,\obs}(x,t)
         \ge \JRAA{x,t}{\tar,\obs}(\wSig, \firstStrat) - \ep \by{\eqref{eqn:raa-ge-first-strat}}
         \\
         &\ge \min\lf\{ \tarEval{\firstTraj{\wSig}}{s}, \min_{\ta \in [t,T]} \obsEval{\firstTraj{\wSig}}{\ta} \rg\} - \ep
         \\
         &= \min\lf\{ \tar(z,s), \min_{\ta \in [t,s]} \obsEval{\firstTraj{\vSig}}{\ta}, \JA{z,s}{\obs}(\bar{\vSig}, \bar{\firstStrat}) \rg\} - \ep
         \by{\eqref{eqn:raa-ge-z-in-terms-of-w}, \eqref{eqn:raa-ge-second-trajectory}}
         \\
         &\ge \min\lf\{ \tar(z,s), \min_{\ta \in [t,s]} \obsEval{\firstTraj{\vSig}}{\ta}, \VA{\obs}(z,s)\rg\} - 2\ep \by{\eqref{eqn:raa-ge-vbar}}
         \\
         &= \min\lf\{ \specialTar(z,s), \min_{\ka \in [t,s]} \obsEval{\firstTraj{\vSig}}{\ka} \rg\} - 2\ep
         \\
         &\ge \VRA{\specialTar,\obs}(x,t) - 3\ep.
        \by{\eqref{eqn:raa-ge-v-and-s}, \eqref{eqn:raa-ge-z}}
    \end{align*}
\end{proof}

\begin{lemma}\label{lemma:easy-lemma}
Let $a,b,c,d,e \in \R$, with $a \ge c$.
Then
$$\min\lf\{\max\{a,e\}, \max\{c,b\}\rg\} 
\ge \max\lf\{\min\{a,b\}, \min\{c,d\}\rg\}.$$
\end{lemma}
\begin{proof}
    Observe that $\max\{a,e\} \ge a \ge \min\{a,b\}$, $\max\{a,e\} \ge a \ge c \ge \min\{c,d\}$, $\max\{c,b\} \ge b \ge \min\{a,b\}$, and $\max\{c,b\} \ge c \ge \min\{c,d\}$.
\end{proof}

\begin{proof}[Proof of Theorem \ref{theorem:rr-theorem}]
    Let $\prm$, $\alt$, and $\specialTar$ be as in the hypothesis (recall $\prm$ and $\alt$ are assumed continuous, so $\specialTar$ is as well).
    We fix $x \in \Rn$ and $t \le T$.
    First observe that
    \begin{align*}
         &\VRR{\prm,\alt}(x,T)
         = \min \{\prm(x,T), \alt(x,T)\}
         \\
         &= \max \{\min\{ \prm(x,T), \alt(x,T)\}, \min\{ \alt(x,T), \prm(x,T)\}\}
         \\
         &= \VR{\specialTar}(x,T),
    \end{align*}
    so it suffices to assume that $t < T$.
    We show \eqref{eqn:rr-theorem} holds.
    
    \noindent($\le$)
    Fix $\ep > 0$.
    Select $\firstStrat \in \dStrats{t}$ such that
	\begin{equation}\label{eqn:rr-le-alpha}
        \textstyle\supu{t} \JR{x,t}{\specialTar}(\uSig, \firstStrat) \le \VR{\specialTar}(x,t) + \ep.
    \end{equation}
    For each $z \in \Rn$ and $s \in [t,T)$, choose $\prmStrat{z,s}, \altStrat{z,s} \in \dStrats{s}$ s.t.
    \begin{align}
        &\textstyle\supu{s} \JR{z,s}{\prm}\lf(\uSig, \prmStrat{z,s} \rg) \le \VR{\prm}(z,s) + \ep, 
        \nonumber
        \\
        &\textstyle\supu{s} \JR{z,s}{\alt}\lf(\uSig, \altStrat{z,s} \rg) \le \VR{\alt}(z,s) + \ep.\label{eqn:rr-le-eta}
    \end{align}

    For each $\uSig \in \uSigs{t}$, define the switch sets
    \begin{align}
        \prmSwitchSet{\uSig}    
	& := \lf\{ s \in [t,T) :  \prmEval{\firstTraj{\uSig}}{s} \ge \VREval{\alt}{\firstTraj{\uSig}}{s} \rg\}, 
	\label{eqn:rr-le-s-switch-set} \\
        \altSwitchSet{\uSig} 
	& := \lf\{ s \in [t,T) : \altEval{\firstTraj{\uSig}}{s} \ge \VREval{\prm}{\firstTraj{\uSig}}{s} \rg\},
	\label{eqn:rr-le-r-switch-set}
    \end{align}
    set $\prmSwitchTime{\uSig} := \min \prmSwitchSet{\uSig}$ if $\prmSwitchSet{\uSig} \ne \varnothing$ and $\prmSwitchTime{\uSig} := T$ otherwise, set $\altSwitchTime{\uSig} := \min \altSwitchSet{\uSig}$ if $\altSwitchSet{\uSig} \ne \varnothing$ and $\altSwitchTime{\uSig} := T$ otherwise, and set $\prmSwitchState{\uSig} := \firstTraj{\uSig}(\prmSwitchTime{\uSig})$ and $\altSwitchState{\uSig} := \firstTraj{\uSig}(\altSwitchTime{\uSig})$.

    We define the composite strategy $\composedStrat \in \dStrats{t}$ by for each $\uSig \in \uSigs{t}$ letting $\composedStrat[\uSig] \in \dSigs{t}$ be given by
	\begin{equation}\label{eqn:rr-le-gamma}
        \composedStrat[\uSig](\ta) := \begin{cases}
            \firstStrat[\uSig](\ta) & \ta < \min\{\prmSwitchTime{\uSig}, \altSwitchTime{\uSig}\} \\
            \altStrat{\prmSwitchState{\uSig}, \prmSwitchTime{\uSig}}[\uSig \vert_{ [\prmSwitchTime{\uSig}, T)}](\ta) & \ta \ge  \prmSwitchTime{\uSig}, \altSwitchTime{\uSig} \ge \prmSwitchTime{\uSig} \\
            \prmStrat{\altSwitchState{\uSig}, \altSwitchTime{\uSig}}[\uSig \vert_{ [\altSwitchTime{\uSig}, T)}](\ta) & \ta \ge \altSwitchTime{\uSig}, \prmSwitchTime{\uSig} > \altSwitchTime{\uSig}.
        \end{cases}
    \end{equation}
    The claim that $\composedStrat$ is non-anticipative, i.e. $\composedStrat \in \dStrats{t}$, is shown similarly to the proof of Theorem \ref{theorem:raa-theorem}; we provide an outline for space.
    Suppose $\uSig_1, \uSig_2 \in \uSigs{t}$ agree a.e. on $[t,b) \sbs [t,T]$.
    The same argument as before shows that if $\min\{\prmSwitchTime{\uSig_1}, \prmSwitchTime{\uSig_2} \} \le b$ then $\prmSwitchTime{\uSig_1} = \prmSwitchTime{\uSig_2}$ and $\prmSwitchState{\uSig_1} = \prmSwitchState{\uSig_2}$, and the analogous result also holds for $\altSwitchTime{\uSig_1}$, $\altSwitchTime{\uSig_2}$, $\altSwitchState{\uSig_1}$,
    $\altSwitchState{\uSig_2}$.
    Thus if $b \le \min\{\prmSwitchTime{\uSig_1}, \altSwitchTime{\uSig_1}\}$, the claim follows by non-anticipativity of $\firstStrat$, so assume otherwise.
    If $\prmSwitchTime{\uSig_1} \le \altSwitchTime{\uSig_1}$, the claim follows from non-anticipativity of $\firstStrat$ and $\altStrat{\prmSwitchState{\uSig_1}, \prmSwitchTime{\uSig_1}}$.
    Otherwise, it follows from non-anticipativity of $\firstStrat$ and $\prmStrat{\altSwitchState{\uSig_1}, \altSwitchTime{\uSig_1}}$.
    
    Now, choose a control signal $\vSig \in \uSigs{t}$ such that
	\begin{equation}\label{eqn:rr-le-v}
        \textstyle\JRR{x,t}{\prm,\alt}(\vSig, \composedStrat) \ge \supu{t} \JRR{x,t}{\prm,\alt}(\uSig, \composedStrat) - \ep.
    \end{equation}
    Let $s_1 := \prmSwitchTime{\vSig}$, $s_2 := \altSwitchTime{\vSig}$, and
	\begin{equation}\label{eqn:rr-le-z}
        z_1 := \prmSwitchState{\vSig} = \firstTraj{\vSig}(s_1),\quad z_2 := \altSwitchState{\vSig} = \firstTraj{\vSig}(s_2).
    \end{equation}

    First assume $s_1 \le s_2$.
    There are two subcases: $s_1 \in [t,T)$ and $s_1 = T$.
    First suppose that $s_1 \in [t,T)$.
    Then by definition of $s_1$ and $s_2$, from \eqref{eqn:rr-le-s-switch-set} and \eqref{eqn:rr-le-r-switch-set},
    for all $\tau \in [t,s_1)$ we have
    \begin{align}
	    \prmEval{\firstTraj{\vSig}}{\ta} & < \VREval{\alt}{\firstTraj{\vSig}}{\ta},
	    \label{eqn:rr-le-switching-condition-l-and-Vh} \\
	    \altEval{\firstTraj{\vSig}}{\ta} & < \VREval{\prm}{\firstTraj{\vSig}}{\ta},
	    \label{eqn:rr-le-switching-condition-h-and-Vl}
    \end{align}
    and also since $s_1 \in \prmSwitchSet{\vSig}$, we have by $\eqref{eqn:rr-le-s-switch-set}$
    \begin{equation}\label{eqn:rr-le-from-switching-condition}
        \prm(z_1,s_1) \ge \VR{\alt}(z_1,s_1).
    \end{equation}
    Let $\bar{\vSig} := \vSig \vert_{[s_1,T)}$.
    By \eqref{eqn:rr-le-gamma}, \eqref{eqn:rr-le-z}, and the assumption $s_1 \le s_2$, for all $\ta \in [t,T]$ we have
    \begin{equation*}
	    \composedStrat[\vSig](\ta) = \begin{cases}
		    \firstStrat[\vSig](\ta) & \ta < s_1,\\
		    \altStrat{z_1,s_1}[\bar{\vSig}](\ta) & \ta \ge s_1.\\
	    \end{cases}
    \end{equation*}
    Then for all $\ta \in [t,s_1]$, 
    \begin{equation}\label{eqn:rr-le-traj-before-s}
	    \composedTraj{\vSig}(\ta) = \firstTraj{\vSig}(\ta),
    \end{equation}
    and for all $\ta \in [s_1,T]$,
    \begin{equation}\label{eqn:rr-le-traj-after-s}
	    \composedTraj{\vSig}(\ta) = \traj{z_1,s_1}{\bar{\vSig}, \altStrat{z_1,s_1}}(\ta).
    \end{equation}
    Thus
    \begin{align*}
        &\VRR{\prm,\alt}(x,t)
        \le \JRR{x,t}{\prm, \alt}(\vSig, \composedStrat) + \ep
	    \by{\eqref{eqn:rr-le-v}}
        \\
        &\le \JR{x,t}{\alt}(\vSig, \composedStrat) + \ep
        \\
        &=\max\lf\{
        \max_{\ta \in [t,s_1]} \altEval{\firstTraj{\vSig}}{\ta},
        \JR{z_1,s_1}{\alt}(\bar{\vSig}, \altStrat{z_1,s_1})
        \rg\} + \ep   
	    \by{\eqref{eqn:rr-le-traj-before-s},\eqref{eqn:rr-le-traj-after-s}}
        \\
        &\le
        \max\lf\{
        \max_{\ta \in [t,s_1]} \altEval{\firstTraj{\vSig}}{\ta},
        \VR{\alt}(z_1,s_1)
	    \rg\} + 2\ep \by{\eqref{eqn:rr-le-eta}}
        \\
        &\le \max \lf\{
        \max_{\ta \in [t,s_1]} \altEval{\firstTraj{\vSig}}{\ta},
        \max_{\ta \in [t,s_1]} \prmEval{\firstTraj{\vSig}}{\ta}
        \rg\} \\
        & \quad + 2\ep     
        \by{\eqref{eqn:rr-le-z},\eqref{eqn:rr-le-from-switching-condition}}
        \\ 
        &= \max_{\ta \in [t,s_1]} \max \lf\{
        \prmEval{\firstTraj{\vSig}}{\ta},
        \altEval{\firstTraj{\vSig}}{\ta}
        \rg\} + 2\ep
        \\
        &= \max_{\ta \in [t,s_1]} \max \big\{
        \min\{\prmEval{\firstTraj{\vSig}}{\ta}, \VREval{\alt}{\firstTraj{\vSig}}{\ta}\},
        \\
        &\qquad\qquad\min\{\altEval{\firstTraj{\vSig}}{\ta}, \VREval{\prm}{\firstTraj{\vSig}}{\ta}\}
        \big\} + 2\ep                                             
        \by{\eqref{eqn:rr-le-switching-condition-l-and-Vh},\eqref{eqn:rr-le-switching-condition-h-and-Vl}}
        \\
        &\le \JR{x,t}{\specialTar}(\vSig, \firstStrat) + 2\ep
        \\
	    &\le \VR{\specialTar}(x,t) + 3\ep. \by{\eqref{eqn:rr-le-alpha}}
    \end{align*}

    Instead, suppose that $s_1 = T$.
    By the assumption that $s_1 \le s_2$, we have $s_2 = T$ as well.
    By \eqref{eqn:rr-le-gamma}, we have
    \begin{equation}\label{eqn:rr-le-T-alpha-is-gamma}
	    \composedStrat = \firstStrat.
    \end{equation}
    Moreover, by definition of $s_1$ and $s_2$ it follows that $\prmSwitchSet{\vSig} = \altSwitchSet{\vSig} = \varnothing$,
    so that by \eqref{eqn:rr-le-s-switch-set} and \eqref{eqn:rr-le-r-switch-set}, for all $\ta \in [t,T]$ we have
    \begin{align}
	    \prmEval{\firstTraj{\vSig}}{\ta} & \le \VREval{\alt}{\firstTraj{\vSig}}{\ta},
	    \label{eqn:rr-le-T-switching-condition-l-Vh} \\
        \altEval{\firstTraj{\vSig}}{\ta} & \le \VREval{\prm}{\firstTraj{\vSig}}{\ta}
	    \label{eqn:rr-le-T-switching-condition-h-Vl}.
    \end{align}
    Thus
    \begin{align*}
        &\VRR{\prm,\alt}(x,t)
        \le \JRR{x,t}{\prm, \alt}(\vSig, \composedStrat) + \ep
	    \by{\eqref{eqn:rr-le-v}}
        \\
	    &= \JRR{x,t}{\prm, \alt}(\vSig, \firstStrat) + \ep
	    \by{\eqref{eqn:rr-le-T-alpha-is-gamma}}
        \\
        &\le \max_{\ta \in [t,T]} \max \lf\{
        \prmEval{\firstTraj{\vSig}}{\ta},
        \altEval{\firstTraj{\vSig}}{\ta}
        \rg\} + \ep
        \\
        &= \max_{\ta \in [t,T]} \max \big\{
        \min\{\prmEval{\firstTraj{\vSig}}{\ta}, \VREval{\alt}{\firstTraj{\vSig}}{\ta}\},
        \\
        & \qquad\qquad\quad\min\{\altEval{\firstTraj{\vSig}}{\ta}, \VREval{\prm}{\firstTraj{\vSig}}{\ta}\}
        \big\} + \ep                                                   
        \by{\eqref{eqn:rr-le-T-switching-condition-l-Vh},\eqref{eqn:rr-le-T-switching-condition-h-Vl}}
        \\
        &= \JR{x,t}{\specialTar}(\vSig, \firstStrat) + \ep 
	    \by{\eqref{eqn:rr-le-alpha}}
        \\
        &\le \VR{\specialTar}(x,t) + 2\ep.
    \end{align*}

    In the case $s_1 > s_2$, we only must consider the possibility that $s_2 \in [t,T)$.
    This argument proceeds identically to the $s_1 \in [t,T)$ subcase of the $s_1 \le s_2$ case, with the variables $s_1,s_2,z_1,z_2,\prm,\alt,\altStrat{z_1,s_1}$ switched with $s_2,s_1,z_2,z_1,\alt,\prm,\prmStrat{z_2,s_2}$, respectively.

    \noindent($\ge$)
    Fix $\ep > 0$.
    First, choose $\firstStrat \in \dStrats{t}$ such that
    \begin{equation}\label{eqn:rr-ge-alpha}
        \textstyle \supu{t} \JRR{x,t}{\prm,\alt}(\uSig, \firstStrat) \le \VRR{\prm,\alt}(x,t) + \ep.
    \end{equation}
    Next, choose $\vSig \in \uSigs{t}$ and $s \in [t,T]$ such that
    \begin{equation}\label{eqn:rr-ge-v-and-s}
        \textstyle \specialTarEval{\firstTraj{\vSig}}{s} > \supu{t} \JR{x,t}{\specialTar}(\uSig, \firstStrat) - \ep.
    \end{equation}
    Since the left-hand-side of $\eqref{eqn:rr-ge-v-and-s}$ is continuous in $s$ and the inequality is strict, it suffices to assume $s < T$.
    Let
    \begin{equation}\label{eqn:rr-ge-z}
        z := \firstTraj{\vSig}(s).
    \end{equation}

    For each $\uSig \in \uSigs{s}$, let $\uSig_\vSig \in \uSigs{t}$ be given by
    \begin{equation*}
        \uSig_\vSig(\ta) := \begin{cases}
            \vSig(\ta) & \ta < s,   \\
            \uSig(\ta) & s \le \ta.
        \end{cases}
    \end{equation*}
    Define the adversary strategy $\bar{\firstStrat}: \uSigs{s} \to \dSigs{s}$ by $\bar{\firstStrat}[\uSig] = \firstStrat[\uSig_\vSig]|_{[s,T)}$.
    Because $\firstStrat$ is non-anticipative, so is $\bar{\firstStrat}$, i.e. $\bar{\firstStrat} \in \dStrats{s}$.
    Without loss of generality, suppose that
    \begin{equation}\label{eqn:rr-ge-l-ge-h}
        \prm(z,s) \ge \alt(z,s).
    \end{equation}
    Select $\bar{\vSig} \in \uSigs{s}$ such that
    \begin{equation}\label{eqn:rr-ge-vbar}
        \textstyle\JR{z,s}{\alt}(\bar{\vSig}, \bar{\firstStrat}) \ge \supu{s} \JR{z,s}{\alt}(\uSig, \bar{\firstStrat}) - \ep.
    \end{equation}
    Define $\wSig \in \uSigs{t}$ by
    \begin{equation*}
        \wSig(\ta) = \begin{cases}
            \vSig(\ta)       & \ta < s,    \\
            \bar{\vSig}(\ta) & s \le \ta.
        \end{cases}
    \end{equation*}
    Note that we have $\bar{\firstStrat}[\bar{\vSig}] = \firstStrat[\wSig] \vert_{[s,T)}$
    and for all $\ta \in [t,s]$,
    \begin{equation}\label{eqn:rr-ge-trajectory-before-s}
	    \firstTraj{\wSig}(\ta) = \firstTraj{\vSig}(\ta).
    \end{equation}
    Thus by $\eqref{eqn:rr-ge-z}$, we have $z = \firstTraj{\wSig}(s)$.
    But then it follows that for all $\ta \in [s,T]$, we have
    \begin{equation}\label{eqn:rr-ge-trajectory-after-s}
        \firstTraj{\wSig}(\ta) = \barTraj(\ta).
    \end{equation}
    Thus
    \begin{align*}
        &\VRR{\prm,\alt}(x,t)
	    \ge \JRR{x,t}{\prm,\alt}(\wSig,\firstStrat) - \ep \by{\eqref{eqn:rr-ge-alpha}}
        \\
        & = \min\bigg\{
        \max\lf\{
        \max_{\ta \in [t,s]} \prmEval{\firstTraj{\vSig}}{\ta},
        \JR{z,s}{\prm}(\bar{\vSig}, \bar{\firstStrat})
        \rg\},
        \\
        & \qquad\qquad\max\lf\{
        \max_{\ta \in [t,s]} \altEval{\firstTraj{\vSig}}{\ta},
        \JR{z,s}{\alt}(\bar{\vSig}, \bar{\firstStrat})
        \rg\}
        \bigg\} - \ep      
        \by{\eqref{eqn:rr-ge-trajectory-before-s},\eqref{eqn:rr-ge-trajectory-after-s}}\\
        & \ge \min\big\{
        \max\lf\{
        \prm(z,s),
        \JR{z,s}{\prm}(\bar{\vSig}, \bar{\firstStrat})
        \rg\},
        \\
        &\qquad\quad~~ \max\lf\{\alt(z,s),
        \JR{z,s}{\alt}(\bar{\vSig}, \bar{\firstStrat})
        \rg\}
        \big\} - \ep
	    \by{\eqref{eqn:rr-ge-z}}
        \\
        & \ge \min\big\{
        \max\lf\{
        \prm(z,s),
        \JR{z,s}{\prm}(\bar{\vSig}, \bar{\firstStrat})
        \rg\},
        \\
        &\qquad\quad~~ \max\lf\{\alt(z,s),
        \VR{\alt}(z,s)
        \rg\}
        \big\} - 2\ep         
	    \by{\eqref{eqn:rr-ge-vbar}}
        \\
        & \ge \max\big\{
        \min\lf\{
        \prm(z,s),
        \VR{\alt}(z,s)
        \rg\},
        \\
        &\qquad\quad~~ \min\lf\{\alt(z,s),
        \VR{\prm}(z,s)
        \rg\}
        \big\}- 2\ep
	    \by{Lemma \ref{lemma:easy-lemma} and \eqref{eqn:rr-ge-l-ge-h}}
        \\
        & = \specialTar(z,s) - 2\ep
        \\
        & \ge \VR{\specialTar}(x,t) - 3\ep. \by{\eqref{eqn:rr-ge-v-and-s}}
    \end{align*}

\end{proof}

\bibliographystyle{ieeetr}
\bibliography{references}

\end{document}